\documentclass[11pt]{article}

\usepackage{amsmath}
\usepackage{amssymb}
\usepackage{amsfonts,amsthm}
\usepackage{thmtools,thm-restate}
\usepackage{array}
\usepackage{multicol}

\newtheorem{theorem}{Theorem}[section]
\newtheorem{lemma}[theorem]{Lemma}

\newtheorem{claim}[theorem]{Claim}

\newtheorem{proposition}[theorem]{Proposition}
\newtheorem{observation}[theorem]{Observation}

\newenvironment{claimproof}[1]{\par\noindent\emph{Proof:}\space#1}{{\leavevmode\unskip\penalty9999 \hbox{}\nobreak\hfill\quad\hbox{$\diamondsuit$}}}

\usepackage{tikz}
\usepackage[a4paper,bindingoffset=0.2in,
            top=1.1in, bottom=1.1in, left=1.0in, right=1.0in,
            footskip=.5in]{geometry}

\newcommand{\ignore}[1]{}

\let\OLDthebibliography\thebibliography
\renewcommand\thebibliography[1]{
  \OLDthebibliography{#1}
  \setlength{\parskip}{2pt}
  \setlength{\itemsep}{3pt plus 0.3ex}
}

\begin{document}

\title{Maximizing the Strong Triadic Closure in Split Graphs and Proper Interval Graphs}

\author{
Athanasios Konstantinidis\thanks{Department of Mathematics, University of Ioannina, Greece. E-mail: \texttt{konsakis@yahoo.com}}
\and
Charis Papadopoulos\thanks{Department of Mathematics, University of Ioannina, Greece. E-mail:  \texttt{charis@cs.uoi.gr}}
}

\date{}

\pagestyle{plain}
\maketitle

\begin{abstract}
In social networks the {\sc Strong Triadic Closure} is an assignment of the edges with strong or weak labels such that any two vertices that have a common neighbor with a strong edge are adjacent.
The problem of maximizing the number of strong edges that satisfy the strong triadic closure was recently shown to be NP-complete for general graphs.
Here we initiate the study of graph classes for which the problem is solvable.
We show that the problem admits a polynomial-time algorithm for two unrelated classes of graphs: proper interval graphs and trivially-perfect graphs.
To complement our result, we show that the problem remains NP-complete on split graphs, and consequently also on chordal graphs.
Thus we contribute to define the first border between graph classes on which the problem is polynomially solvable and on which it remains NP-complete.
\end{abstract}

\section{Introduction}
Predicting the behavior of a network is an important concept in the field of social networks \cite{EK10}.
Understanding the strength and nature of social relationships has found an increasing usefulness in the last years due to the explosive growth of social networks (see e.g., \cite{BK14}). 
Towards such a direction the {\sc Strong Triadic Closure} principle enables to understand the structural properties of the underlying graph:
it is not possible for two individuals to have a strong relationship with a common friend and not know each other \cite{Go73}.
Such a principle stipulates that if two people in a social network have a ``strong friend'' in common, then there is an increased likelihood that they will become friends themselves at some point in the future.
Satisfying the {\sc Strong Triadic Closure} is to characterize the edges of the underlying graph into weak and strong such that any two vertices that have a strong neighbour in common are adjacent.
Since users interact and actively engage in social networks by creating strong relationships, it is natural to consider the {\sc MaxSTC} problem: maximize the number of strong edges that satisfy the {\sc Strong Triadic Closure}.
The problem has been shown to be NP-complete for general graphs whereas its dual problem of minimizing the number of weak edges admits a constant factor approximation ratio \cite{ST14}.

In this work we initiate the computational complexity study of the {\sc MaxSTC} problem in important classes of graphs.
If the input graph is a $P_3$-free graph (i.e., a graph having no induced path on three vertices which is equivalent with a graph that consists of vertex-disjoint union of cliques) then there is a trivial solution by labeling strong all the edges.
Such an observation might falsely lead into a graph modification problem, known as {\sc Cluster Deletion} problem (see e.g., \cite{BDM15,HLNPT10}), in which we want to remove the minimum number of edges that correspond to the weak edges, such that the resulting graph does not contain a $P_3$ as an induced subgraph.
More precisely the obvious reduction would consist in labeling the deleted edges in the instance of {\sc Cluster Deletion} as weak, and the remaining ones as strong.
However, this reduction fails to be correct due to the fact that the graph obtained by deleting the weak edges in an optimal solution of {\sc MaxSTC} may contain
an induced $P_3$, so long as those three vertices induce a triangle in the original graph (prior to deleting the weak edges).

To the best of our knowledge, no previous results were known prior to our work when restricting the input graph for the {\sc MaxSTC} problem.
It is not difficult to see that for bipartite graphs the {\sc MaxSTC} problem has a simple polynomial-time solution by considering a maximum matching that represent the strong edges \cite{LR88}.
It is well-known that a maximum matching of a graph corresponds to a maximum independent set of its line graph that represents the adjacencies between the edges \cite{Edmonds65}.
As previously noted, for general graphs it is not necessarily the case that a maximum matching corresponds to the optimal solution for {\sc MaxSTC}.
Here we show a similar characterization for {\sc MaxSTC} by considering the adjacencies between the edges of a graph that participate in induced $P_3$'s.
Such a characterization allows us to exhibit structural properties towards the computation of an optimal solution of {\sc MaxSTC}.

Due to the nature of the $P_3$ existence that enforce the labeling of weak edges, there is an interesting connection to problems related to the {\it square root} of a graph;
a graph $H$ is a {\it square root} of a graph $G$ and $G$ is the {\it square} of $H$ if two vertices are adjacent in $G$ whenever they are at distance one or two in $H$.
Any graph does not have a square root (for example consider a simple path), but every graph contains a subgraph that has a square root.
Although it is NP-complete to determine if a given chordal graph has a square root \cite{LC04}, there are polynomial-time algorithms when the input is restricted to bipartite graphs \cite{Lau06}, or proper interval graphs \cite{LC04}, or trivially-perfect graphs \cite{MS13}.
The relationship between {\sc MaxSTC} and to that of determining square roots can be seen as follows.
In the {\sc MaxSTC} problem we are given a graph $G$ and we want to select the maximum possible number of edges, at most one from each induced $P_3$ in $G$.
Thus we need to find the largest subgraph (in terms of the number of its edges) $H$ of $G$ such that the square of $H$ is a subgraph of $G$.
However previous results related to square roots were concerned with deciding if the whole graph has a square root and there are no such equivalent formulations related to the largest square root.

Our main motivation is to understand the complexity of the problem on subclasses of chordal graphs, since the class of chordal graphs (i.e., graphs having no chordless cycle of length at least four) finds important applications in both theoretical and practical areas related to social networks \cite{ASM16,KleitmanV90,Pfaltz13}.
More precisely two famous properties can be found in social networks.
For most known social and biological networks their diameter, that is, the length of the longest shortest path between any two vertices of a graph, is known to be a small constant \cite{Jackson08}.
On the other hand it has been shown that the most prominent social network subgraphs are cliques, whereas highly infrequent induced subgraphs are cycles of length four \cite{UBK13}.
Thus it is evident that subclasses of chordal graphs are close related to such networks, since they have rather small diameter (e.g., split graphs or trivially-perfect graphs) and are characterized by the absence of chordless cycles (e.g., proper interval graphs).
Towards such a direction we show that {\sc MaxSTC} is NP-complete on split graphs and consequently also on chordal graphs.
On the positive side, we present the first polynomial-time algorithm for computing {\sc MaxSTC} on proper interval graphs.
Proper interval graphs, also known as unit interval graphs or indifference graphs, form a subclass of interval graphs and they are unrelated to split graphs \cite{R69}.
By our result they form the first graph class, other than bipartite graphs, for which {\sc MaxSTC} is shown to be polynomial time solvable.
In order to obtain our algorithm, we take advantage of their clique path and apply a dynamic programming on subproblems defined by passing the clique path in its natural ordering.
Furthermore by considering the equivalent transformation of the problem mentioned earlier, we show that {\sc MaxSTC} admits a simple polynomial-time solution on trivially-perfect graphs.
Thus we contribute to define the first borderline between graph classes on which the problem is polynomially solvable and on which it remains NP-complete.


\section{Preliminaries}
All graphs considered here are simple and undirected.
A graph is denoted by $G=(V,E)$ with vertex set $V$ and edge set $E$. We use
the convention that $n=|V|$ and $m=|E|$.
The {\it neighborhood} of
a vertex~$v$ of $G$ is $N(v)=\{x \mid vx \in E\}$ and the
{\it closed neighborhood} of $v$ is $N[v] = N(v) \cup \{v\}$.
For $S \subseteq V$, $N(S)=\bigcup_{v \in S} N(v) \setminus S$ and $N[S] = N(S) \cup S$.
A graph~$H$ is a {\it subgraph} of $G$ if $V(H)\subseteq V(G)$
and $E(H)\subseteq E(G)$. For $X\subseteq V(G)$, the subgraph
of $G$ {\it induced} by $X$, $G[X]$, has vertex set~$X$, and
for each vertex pair~$u, v$ from $X$, $uv$ is an edge of $G[X]$
if and only if $u\not= v$ and $uv$ is an edge of $G$.
For $R\subseteq E(G)$, $G\setminus R$ denotes the
graph~$(V(G), E(G)\setminus R)$, that is a subgraph of $G$ and
for $S \subseteq V(G)$, $G - S$ denotes the
graph~$G[V(G)-S]$, that is an induced subgraph of $G$.
Two adjacent vertices $u$ and $v$ are called {\it twins} if $N[u] = N[v]$.

A {\it clique} of $G$ is a set of pairwise adjacent vertices
of $G$, and a {\it maximal clique} of $G$ is a clique of $G$
that is not properly contained in any clique of $G$. An
{\it independent set} of $G$ is a set of pairwise non-adjacent
vertices of $G$.
For $k\geq 2$, the chordless path on $k$ vertices is denoted by $P_k$ and the
chordless cycle on $k$ vertices is denoted by $C_k$.
For an induced path $P_k$, the vertices of degree one are called endvertices.
A vertex $v$ is {\it universal} in $G$ if $N[v] = V(G)$ and $v$ is {\it isolated} if $N(v) = \emptyset$.
For two vertices $u$ and $v$ we say that $u$ {\it sees} $v$ if $\{u,v\}\in E(G)$; otherwise, we say that $u$ \emph{misses} $v$.
We extend this notion to vertex sets: a set $A$ sees (resp., misses) a vertex set $B$ if every vertex of $A$ is adjacent (resp., non-adjacent) to every vertex of $B$.
We say that two {\it edges are non-adjacent} if they have no common endpoint; otherwise we call them {\it adjacent edges}.

\begin{figure}[t]
\centering
\includegraphics[scale= 1.0]{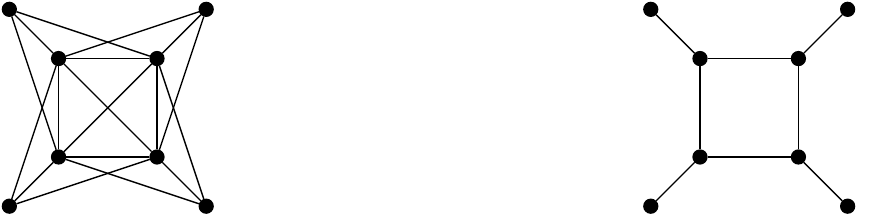}
\caption{A split graph $G$ is shown to the left side.
The right side depicts a solution for {\sc MaxSTC} on $G$ where the weak edges are exactly the edges of $G$ that are not shown.}
\label{fig:split}
\end{figure}

\paragraph{Strong Triadic Closure.}
Given a graph $G = (V,E)$, a {\it strong-weak labeling} on the edges of $G$ is a function $\lambda$
that assigns to each edge of $E(G)$ one of the labels {\it strong} or {\it weak}; i.e., $\lambda: E(G) \rightarrow \{\rm{strong},\rm{weak}\}$.
An edge that is labeled strong (resp., weak) is simple called {\it strong} (resp. {\it weak}).
The {\it strong triadic closure} of a graph $G$ is a strong-weak labeling $\lambda$ such that
for any two strong edges $\{u,v\}$ and $\{v,w\}$ there is a (weak or strong) edge $\{u,w\}$.
In other words, in a strong triadic closure there are no pair of strong edges $\{u,v\}$ and $\{v,w\}$ such that $\{u,w\} \notin E(G)$.

The problem of computing the maximum strong triadic closure, denoted by {\sc MaxSTC}, is to find a strong-weak labeling on the edges of $E(G)$ that satisfies the strong triadic closure and has the maximum number of strong edges.
Note that its dual problem asks for the minimum number of weak edges.
Here we focus on maximizing the number of strong edges in a strong triadic closure.

Let $G$ be a strong-weak labeled graph. We denote by $(E_S,E_W)$ the partition of $E(G)$ into strong edges $E_S$ and weak edges $E_W$.
The graph spanned by $E_S$ is the graph $G \setminus E_W$.
For a vertex $v \in V(G)$ we say that the \emph{strong neighbour} of $v$ is the other endpoint of a strong edge incident to $v$.
We denote by $N_S(v) \subseteq N(v)$ the strong neighbours of $v$.
Similarly we say that a vertex $u$ is \emph{strongly adjacent} to $v$ if $u$ is adjacent to $v$ and the edge $\{u,v\}$ is strong.

\begin{observation}\label{obs:stcP3Clique}
Let $G$ be a strong-weak labeled graph. Then $G \setminus E_W$ satisfies the strong triadic closure if and only if for every $P_3$ in $G \setminus E_W$, the vertices of the $P_3$ induce a $K_3$ in $G$.
\end{observation}
\begin{proof}
Observe that $G \setminus E_W$ is the graph spanned by the strong edges.
If for two strong edges $\{u,v\}$ and $\{v,w\}$, $\{u,w\} \notin E(G \setminus E_W)$ then $\{u,w\}$ is an edge in $G$ and, thus, $u,v,w$ induce a $K_3$ in $G$.
On the other hand notice that any two strong edges of $G \setminus E_W$ are either non-adjacent or share a common vertex.
If they share a common vertex then the vertices must induce a $K_3$ in $G$.
\end{proof}

Therefore in the {\sc MaxSTC} problem we want to minimize the number of the removal (weak) edges $E_W$ from $G$
such that every three vertices that induce a $P_3$ in $G \setminus E_W$ form a clique in $G$.
Then it is not difficult to see that $G \setminus E_W$ satisfies the strong triadic closure if and only if for every vertex $v$, $N_S[v]$ induces a clique in $G$.



\section{MaxSTC on split graphs}
Here we provide an NP-hardness result for {\sc MaxSTC} on split graphs.
A graph $G=(V,E)$ is a {\it split graph} if $V$ can be partitioned into a
clique $C$ and an independent set $I$, where $(C,I)$ is called a
{\it split partition} of $G$.
Split graphs form a subclass of the larger and
widely known graph class of {\it chordal graphs}, which are the graphs that do not contain
induced cycles of length 4 or more as induced subgraphs.
It is known that split graphs are self-complementary, that is, the complement of a split graph remains a split graph.
First we show the following result.

\begin{lemma}\label{lem:simplesplitcases}
Let $G=(V,E)$ be a split graph with a split partition $(C,I)$. Let $E_S$ be the set of strong edges in an optimal solution for {\sc MaxSTC} on $G$ and let $I_W$ be the vertices of $I$ that are incident to at least one edge of $E_S$.
\begin{enumerate}
\item If every vertex of $I_W$ misses at least three vertices of $C$ in $G$ then $E_S = E(C)$.
\item If every vertex of $I_W$ misses exactly one vertex of $C$ in $G$ then $|E_S| \leq |E(C)| + \lfloor\frac{|I_W|}{2}\rfloor$.
\end{enumerate}
\end{lemma}
\begin{proof}
Let $w_i$ be a vertex of $I$ and let $B_i$ be the set of vertices in $C$ that are non-adjacent to $w_i$.
Let $A_i$ be the strong neighbors of $w_i$ in an optimal solution.
For the edges of the clique, there are $|A_i||B_i|$ weak edges due to the strong triadic closure.
Moreover any vertex $w_j$ of $I \setminus \{w_i\}$ cannot have a strong neighbor in $A_i$.
This means that $A_i \cap A_j = \emptyset$.
Notice, however, that both sets $B_i \cap B_j$ and $A_i \cap B_j$ are not necessarily empty.

Observe that $I_W$ contains the vertices of $I$ that are incident to at least one strong edge.
Let $E(A,B)$ be the set of weak edges that have one endpoint in $A_i$ and the other endpoint in $B_i$, for every $1 \leq i \leq |I_W|$.
We show that $2|E(A,B)| \geq \sum_{w_i \in I_W}|A_i||B_i|$.
Let $\{a,b\}\in E(A,B)$ such that $a\in A_i$ and $b\in B_i$.
Assume that there is a pair $A_j,B_j$ such that $\{a,b\}$ is an edge between $A_j$ and $B_j$, for $j \neq i$.
Then $a$ cannot belong to $A_j$ since $A_i \cap A_j = \emptyset$.
Thus $a\in A_i$ and $b\in A_j$.
Therefore for every edge $\{a,b\}\in E(A,B)$ there are at most two pairs $(A_i,B_i)$ and $(A_j,B_j)$ for which $a \in A_i \cup B_j$ and $b\in B_i \cup A_j$.
This means that every edge of $E(A,B)$ is counted at most twice in $\sum_{w_i \in I_W}|A_i||B_i|$.

For any two edges $\{u,v\}, \{v,z\}\in E(C) \setminus E(A,B)$, observe that they satisfy the strong triadic closure since there is the edge $\{u,z\}$ in $G$.
Thus the strong edges of the clique are exactly the set of edges $E(C) \setminus E(A,B)$.
In total by counting the number of strong edges between the independent set and the clique, we have $|E_S| = |E(C) \setminus E(A,B)| + \sum_{w_i \in I_W} |A_i|$.
Since $2|E(A,B)| \geq \sum_{w_i \in I_W}|A_i||B_i|$, we get
$$
|E_S| \leq |E(C)| + \sum_{w_i \in I_W}|A_i|\left(1 - \left\lfloor\frac{|B_i|}{2}\right\rfloor\right).
$$

Now the first claim of the lemma holds because $|B_i|=3$ so that $I_W = \emptyset$. 
For the second claim we show that for every vertex of $I_W$, $|A_i|=1$.
Let $w_i \in I_W$ such that $|A_i| \geq 2$ and let $B_i = \{b_i\}$.
Recall that no other vertex of $I_W$ has strong neighbours in $A_i$.
Also note that there is at most one vertex $w_j$ in $I_W$ that has $b_i$ as a strong neighbour.
If such a vertex $w_j$ exist and for the vertex $b_j$ of the clique that misses $w_j$ it holds $b_j \in A_i$, then we let $v = b_j$; otherwise we choose $v$ as an arbitrary vertex of $A_i$.
Observe that no vertex of $I \setminus \{w_i\}$ has a strong neighbour in $A_i \setminus \{v\}$ and only $w_j \in I_W$ is strongly adjacent to $b_i$.
Then we label weak the $|A_i|-1$ edges between $w_i$ and the vertices of $A_i \setminus \{v\}$ and we label strong the $|A_i|-1$ edges between $b_i$ and the vertices of $A_i \setminus \{v\}$.
Making strong the edges between $b_i$ and the vertices of $A_i \setminus \{v\}$ does not violate the strong triadic closure since every vertex of $C \cup \{w_j\}$ is adjacent to every vertex of $A_i \setminus \{v\}$.
Therefore for every vertex $w_i \in I_W$, $|A_i|=1$ and by substituting $|B_i|=1$ in the formula for $|E_S|$ we get the claimed bound.
\end{proof}

In order to give the reduction, we introduce the following problem that we call {\it maximum disjoint non-neighbourhood}:
given a split graph $(C,I)$ where every vertex of $I$ misses three vertices from $C$,
we want to find the maximum subset $S_I$ of $I$ such that the non-neighbourhoods of the vertices of $S_I$
are pairwise disjoint.
In the corresponding decision version, denoted by {\sc MaxDisjointNN}, we are also given an integer $k$ and the problem asks whether $|S_I| \geq k$.
The polynomial-time reduction to {\sc MaxDisjointNN} is given from the classical NP-complete problem {\sc 3-Set Packing} \cite{Karp72}: given a universe $\mathcal{U}$ of $n$ elements, a family $\mathcal{F}$ of triplets of $\mathcal{U}$, and an integer $k$, the problem asks for
a subfamily $\mathcal{F'} \subseteq \mathcal{F}$ with $|\mathcal{F'}| \geq k$ such that all triplets of $\mathcal{F'}$ are pairwise disjoint.

\begin{theorem}\label{theo:mndsplithard}
{\sc MaxDisjointNN} is NP-complete on split graphs.
\end{theorem}
\begin{proof}
Given a split graph $G=(C,I)$ and $S_I \subseteq I$, checking whether $S_I$ is a solution for {\sc MaxDisjointNN} amounts to checking whether every pair of vertices of $S_I$ have common neighbourhood. As this can be done in polynomial time the problem is in NP. We will give a polynomial-time reduction to {\sc MaxDisjointNN} from the classical NP-complete problem {\sc 3-Set Packing} \cite{Karp72}: given a universe $\mathcal{U}$ of $n$ elements, a family $\mathcal{F}$ of triplets of $\mathcal{U}$, and an integer $k$, the problem asks for
a subfamily $\mathcal{F'} \subseteq \mathcal{F}$ with $|\mathcal{F'}| \geq k$ such that all triplets of $\mathcal{F'}$ are pairwise disjoint.

Let $(\mathcal{U},\mathcal{F},k)$ be an instance of the {\sc 3-Set Packing}.
We construct a split graph $G=(C,I)$ as follows.
The clique of $G$ is formed by the $n$ elements of $\mathcal{U}$. For every triplet $F_i$ of $\mathcal{F}$ we add a vertex $v_i$ in $I$ that is adjacent to every vertex of $C$ except the three vertices that correspond to the triplet $F_i$. Thus every vertex of $I$ misses exactly three vertices from $C$ and sees the rest of $C$.
Now it is not difficult to see that there is a solution $\mathcal{F'}$ for {\sc 3-Set Packing}$(\mathcal{U},\mathcal{F},k)$ of size at least $k$ if and only if there is a solution $S_I$ for {\sc MaxDisjointNN}$(G,k)$ of size at least $k$.
For every pair $(F_i, F_j)$ of $\mathcal{F'}$ we know that $F_i \cap F_j = \emptyset$ which implies that the vertices $v_i$ and $v_j$ have disjoint non-neighbourhood since $F_i$ corresponds to the non-neighbourhood of $v_i$.
By the one-to-one mapping between the sets of $\mathcal{F}$ and the vertices of $I$, every set $F_i$ belongs to $\mathcal{F'}$ if and only if $v_i$ belongs to $S_I$.
\end{proof}

Now we turn to our original problem {\sc MaxSTC}. The decision version of {\sc MaxSTC} takes as input a graph $G$ and an integer $k$ and asks whether there is strong-weak labeling of the edges of $G$ that satisfies the strong triadic closure with at least $k$ strong edges.

\begin{theorem}\label{theo:stcsplithard}
The decision version of {\sc MaxSTC} is NP-complete on split graphs.
\end{theorem}
\begin{proof}
Given a strong-weak labeling $(E_S, E_W)$ of a split graph $G=(C,I)$, checking whether $(E_S, E_W)$ satisfies the strong triadic closure amounts to check in $G - E_W$ whether there is a non-edge between its endvertices of every $P_3$ according to Observation~\ref{obs:stcP3Clique}.
Thus by listing all $P_3$'s of $G \setminus E_W$ the problem belongs to NP.
Next we give a polynomial-time reduction to {\sc MaxSTC} from the {\sc MaxDisjointNN} problem on split graphs which is NP-complete by Theorem~\ref{theo:mndsplithard}.
Let $(G,k)$ be an instance of {\sc MaxDisjointNN} where $G = (C, I)$ is a split graph such that every vertex of the independent set $I$ misses exactly three vertices from the clique $C$.
For a vertex $w_i \in I$, we denote by $B_i$ the set of the three vertices in $C$ that are non-adjacent to $w_i$. Let $|C| = n$.
We extend $G$ and construct another split graph $G'$ as follows (see Figure~\ref{fig:splitreduction}):
\begin{itemize}
\item We add $n$ vertices $y_1, \ldots, y_n$ in the clique that consist the set $C_Y$.
\item We add $n$ vertices $x_1, \ldots, x_n$ in the independent set that consist the set $I_X$.
\item For every $1 \leq i \leq n$, $y_i$ is adjacent to all vertices of $(C \cup C_Y \cup I \cup I_X) \setminus \{x_i\}$.
\item For every $1 \leq i \leq n$, $x_i$ is adjacent to all vertices of $(C \cup C_Y)\setminus \{y_i\}$.
\end{itemize}
Thus $w_i$ misses only the vertices of $B_i$ from the clique.
By construction it is clear that $G'$ is a split graph with a split partition $(C \cup C_Y, I \cup I_X)$.
Notice that the clique $C \cup C_Y$ has $2n$ vertices and $G=G'[I \cup C]$.

\begin{figure}[t]
\centering
\includegraphics[scale=1.05]{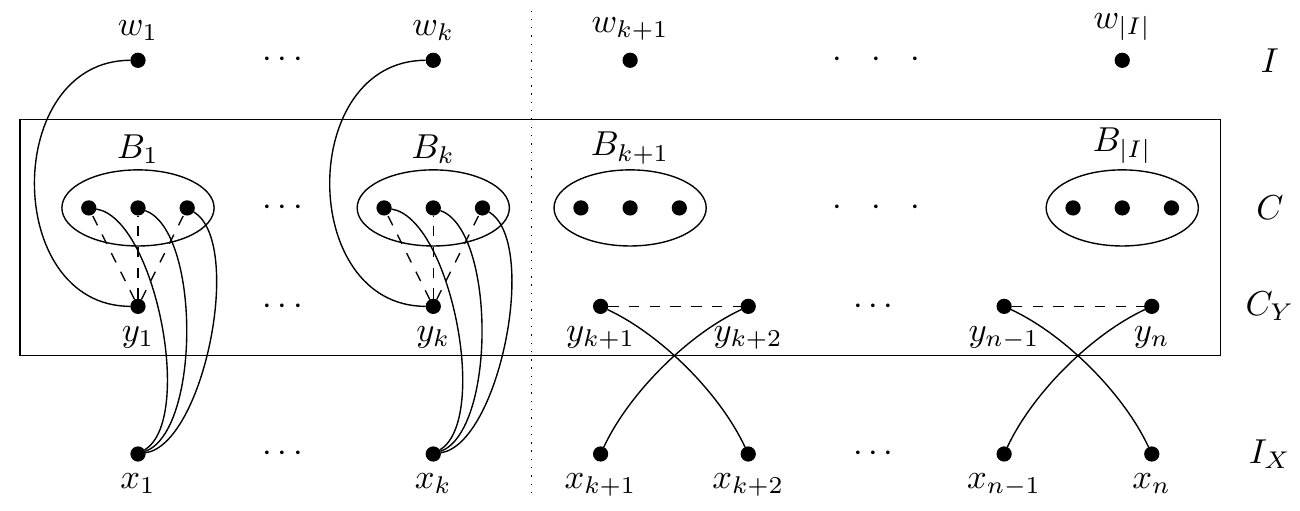}
\caption{The split graph $(C \cup C_Y, I \cup I_X)$ given in the polynomial-time reduction.
Every vertex $w_i$ misses the vertices of $B_i$ and sees the vertices of $(C \cup C_Y) - B_i$.
Every vertex $x_i$ misses $y_i$ and sees the vertices of $(C \cup C_Y) - \{y_i\}$.
The sets $B_1, \ldots, B_k$ are pairwise disjoint whereas for every set $B_j$, $k < j \leq |I|$, there is a set $B_i$, $1 \leq i \leq k$, such that $B_i \cap B_j \neq \emptyset$.
The drawn edges correspond to the strong edges between the independent set and the clique, and the dashed edges are the only weak edges in the clique $C \cup C_Y$.}
\label{fig:splitreduction}
\end{figure}

We claim that $G$ has a solution for {\sc MaxDisjointNN} of size at least $k$ if and only if $G'$ has a strong triadic closure with at least $n(2n-1)+\lfloor\frac{n}{2}\rfloor+\lceil\frac{k}{2}\rceil$ strong edges.

Assume that $\{w_1, \ldots, w_k\} \subseteq I$ is a solution for {\sc MaxDisjointNN} on $G$ of size at least $k$.
Since the sets $B_1, \ldots, B_k$ are pairwise disjoint, there are $k$ distinct vertices $y_1, \ldots, y_k$ in $C_Y$ such that $k \leq n$.
We will give a strong-weak labeling for the edges of $G'$ that fulfills the strong triadic closure and has at least the claimed number of strong edges.
For simplicity, we describe only the strong edges; the edges of $G'$ that are not given are all labeled weak.
We label the edges incident to each vertex $w_i, y_i, x_i$ and the three vertices of each set $B_i$, for $1 \leq i \leq k$ as follows:
\begin{itemize}
\item The edges of the form $\{y_i,v\}$ are labeled strong if $v \in (C \cup C_Y)\setminus B_i$ or $v = w_i$. 
\item The edges incident to $x_i$ and the three vertices of $B_i$ are labeled strong.
\end{itemize}
Next we label the edges incident to the rest of the vertices.
Let $I_W$ be the vertices of $I \setminus \{w_1, \ldots, w_k\}$ and let $C_W$ be the vertices of $C \setminus (B_1 \cup \cdots \cup B_k)$.
No edge incident to a vertex of $I_W$ is labeled strong.
For every vertex $u \in C_W$ we label the edge $\{u,v\}$ strong if $v \in (C \cup C_Y)$.
Let $C'_Y = \{y_{k+1}, \ldots, y_n\}$ and let $I'_X = \{x_{k+1}, \ldots, x_n\}$.
Recall that every vertex $x_{k+j}$ is adjacent to every vertex of $C'_Y \setminus \{y_{k+j}\}$.
Let $\ell = \lfloor\frac{n-k}{2}\rfloor$.
Let $M = \{e_1, \ldots, e_{\ell}\}$ be a maximal set of pairwise non-adjacent edges in $G'[C'_Y]$ where $e_{j} = \{y_{k+2j-1}, y_{k+2j}\}$, for $j \in \{1, \ldots, \ell\}$;
note that $M$ is a maximal matching of $G'[C'_Y]$.
For every vertex $y \in C'_Y$, we label the edge $\{y,v\}$ strong if $v \in (C \cup C_Y)\setminus \{y'\}$ such that $\{y,y'\} \in M$.
Moreover, for $j \in \{1, \ldots, \ell\}$, the edges $\{x_{k+2j-1}, y_{k+2j}\}$ and $\{x_{k+2j}, y_{k+2j-1}\}$ are labeled strong.
Note that if $n-k$ is odd then no edge incident to the unique vertex $y_n$ belongs to $M$ and all edges between $y_n$ and the vertices of $C \cup C_Y$ are labeled strong;
in such a case also note that no edge incident to $x_n$ is strong.

Let us show that such a labeling fulfills the strong triadic closure.
Any labeling for the edges inside $G'[C \cup C_Y]$ is satisfied since $G'[C \cup C_Y]$ is a clique.
Also note that there are no two adjacent strong edges that have a common endpoint in the clique $C \cup C_Y$ and the two other endpoints in the independent set $I \cup I_X$.
If there are two strong edges incident to the same vertex $v$ of the independent set then $v\in \{x_1, \ldots, x_k\}$ and $N_S[v] = B_i$ is a clique.
Assume that there are two adjacent strong edges $\{u,v\}$ and $\{v,z\}$ such that $u \in I \cup I_X$, and $v,z \in C \cup C_Y$.
\begin{itemize}
\item If $u \in \{w_1, \ldots, w_k\}$ then $\{u,z\}\in E(G')$ since every $w_i$ misses only the vertices of $B_i$.
\item If $u \in \{x_1, \ldots, x_k\}$ then $v \in B_i$ and $\{u,z\}\in E(G')$ since every vertex $x_i$ misses only $y_i$.
\item If $u \in I_X \setminus \{x_1, \ldots, x_k\}$ then the strong neighbours of $v$ in $C \cup C_Y$ are adjacent to $u$ in $G'$ since for the only non-neighbour of $u$ in $C \cup C_Y$ there is a weak edge incident to $v$.
\end{itemize}
Recall that there is no strong edge incident to the vertices of $I \setminus \{w_1, \ldots, w_k\}$.
Therefore the given strong-weak labeling fulfills the strong triadic closure.

Observe that the number of vertices in $C \cup C_Y$ is $2n$.
There are exactly $3k+\ell$ weak edges in $G'[C \cup C_Y]$.
Thus the number of strong edges in $G'[C \cup C_Y]$ are $n(2n-1)-3k - \ell$.
There are $k$ strong edges incident to $\{w_1, \ldots, w_k\}$,
$3k$ strong edges incident to $\{x_1, \ldots, x_k\}$, and $2\ell$ strong edges incident to the vertices of $I_X \setminus \{x_1, \ldots, x_k\}$.
Thus the total number of strong edges is $n(2n-1)-3k - \ell + k + 3k + 2\ell = n(2n-1) + \ell + k$ and by substituting $\ell = \lfloor\frac{n-k}{2}\rfloor$ we get the claimed bound.

For the opposite direction, assume that $G'$ has a strong triadic closure with at least $n(2n-1)+\lfloor\frac{n}{2}\rfloor+\lceil\frac{k}{2}\rceil$ strong edges.
Let $E_S$ be the set of strong edges in such a strong-weak labeling.
Observe that the number of edges in $G'[C \cup C_Y]$ is $n(2n-1)$ which implies that $E_S$ contains edges between the independent set $I \cup I_X$ and the clique $C \cup C_Y$.
If no vertex of $I_X$ is incident to an edge of $E_S$ then the first statement of Lemma~\ref{lem:simplesplitcases} implies that $|E_S| =  |E(C\cup C_Y)| = n(2n-1)$.
And if no vertex of $I$ is incident to an edge of $E_S$ then the second statement of Lemma~\ref{lem:simplesplitcases} shows that $|E_S| \leq |E(C\cup C_Y)| + \lfloor\frac{n}{2}\rfloor$.
Therefore $E_S$ contains edges that are incident to a vertex of $I$ and edges that are incident to a vertex of $I_X$.

In the graph spanned by $E_S$ we denote by $S_W$ the set of vertices of $I$ that have strong neighbours in $C \cup C_Y$.
We will show that the non-neighbourhoods of the vertices of $S_W$ in $C \cup C_Y$ are disjoint in $G'$ and, since $G$ is an induced subgraph of $G'$, their non-neighbourhoods are also disjoint in $G$.

\begin{claim}\label{claim:insideCYandx}
For every $w_i \in S_W$, $N_S(w_i) \subseteq C_Y$ and there exists a unique vertex $x \in I_X$ such that $N_S(x) = B_i$.
\end{claim}
\begin{claimproof}
Let $w_i$ be a vertex of $S_W$. We first show that $N_S(w_i) \subseteq C_Y$.
Let $W_i$ be the strong neighbours of $w_i$ in $C$ and let $Y_i$ be the strong neighbours of $w_i$ in $C_Y$.
Observe that no other vertex of $S_W$ has a strong neighbour in $W_i \cup Y_i$.
Further notice that there are $(|W_i|+|Y_i|) |B_i|$ weak edges since $w_i$ is non-adjacent to the vertices of $B_i$.
We show that for every vertex $w_i \in S_W$ it holds $W_i = \emptyset$.
For all vertices $w_i$ for which $W_i \neq \emptyset$ we replace in $E_S$ the strong edges between $w_i$ and the vertices of $W_i$ by the edges between the vertices of $B_i$ and $W_i$.
Notice that making strong the edges between the vertices of $B_i$ and $W_i$ does not violate the strong triadic closure since no vertex from $S_W$ has a strong neighbour in $B_i$ and every vertex of $I_X$ is adjacent to all the vertices of $W_i$.
Let $E(W,B)$ be the set of edges that have one endpoint in $W_i$ and the other endpoint in $B_i$, for every $w_i \in S_W$.
Notice that the difference between the two described solutions is $|E(W,B)| - \sum |W_i|$.
By Lemma~\ref{lem:simplesplitcases} and $|B_i|=3$, we know that $|E(W,B)| \geq 3/2 \sum |W_i|$.
Thus such a replacement is safe for the number of edges of $E_S$ and every vertex $w_i \in S_W$ has strong neighbours only in $C_Y$.

Let $X_i$ be the set of vertices of $I_X$ that have at least one non-neighbour in $Y_i$.
By construction every vertex of $Y_i$ is non-adjacent to exactly one vertex of $I_X$, and thus $|X_i| = |Y_i|$.
Since $w_i$ has strong neighbours in $Y_i$, every edge between $X_i$ and $Y_i$ is weak.
By the previous argument every vertex of $S_W$ has strong neighbours only in $C_Y$ so that $N_S(B_i) \cap I = \emptyset$. 
Also notice that no two vertices of the independent set have a common strong neighbour in the clique, which means that there are at most $|B_i|$ strong neighbours between the vertices of $B_i$ and $I_X$. 
Choose an arbitrary vertex $x \in X_i$.
We replace all strong edges in $E_S$ between $B_i$ and $I_X$ by $|B_i|$ strong edges between $x$ and the vertices of $B_i$.
Notice that such a replacement is safe since the unique non-neighbour of $x$ belongs to $Y_i$ and there are weak edges already in the solution between $B_i$ and $Y_i$ because of the strong edges between $w_i$ and $Y_i$.
Thus $B_i \subseteq N_S(x)$.
We focus on the edges between the vertices of $(C \cup C_Y) \setminus (B_i \cup Y_i)$ and $x$.
If a vertex $x$ of $X_i$ has a strong neighbour $u$ in $(C \cup C_Y) \setminus B_i $ then the edge $\{u,y\}$ is weak where $y \in Y_i$ is the unique non-neighbour of $x$.
Also notice that $N_S(u) \cap (I \cup I_X) = \{x\}$, $N_S(y) \cap (I \cup I_X) = \{w_i\}$, and $w_i$ is adjacent to $u$.
Then we can safely replace the strong edge $\{x,u\}$ by the edge $\{u,y\}$ and keep the same size of $E_S$.
Hence $N_S(x) = B_i$.
\end{claimproof}

\begin{claim}\label{claim:yCy}
For every $w_i \in S_W$, $N_S(w_i) = \{y\}$ where $y \in C_Y$ is the non-neighbour of $x$ with $N_S(x) = B_i$. 
\end{claim}
\begin{claimproof}
Let $Y_i = N_S(w_i)$.
By Claim~\ref{claim:insideCYandx} we know that $Y_i \subseteq C_Y$ and there exists $x \in I_X$ such that $N_S(x) = B_i$.
Both $w_i$ and $x$ are vertices of the independent set and, thus, no other vertex of $I \cup I_X$ has strong neighbours in $B_i \cup Y_i$.
This means that if we remove $w_i$ from $S_W$ by making weak the edges incident to $w_i$ and the vertices of $Y_i$ then the edges between the vertices of $B_i$ and $Y_i \setminus \{y\}$ are safely turned into strong.
Let $E'_S$ be the set of strong edges in an optimal solution such that all edges incident to $w_i$ are weak.
Then $|E_S| - |E'_S| = |Y_i| + |B_i| - |Y_i||B_i|$ and $|E_S| > |E'_S|$ only if $|Y_i| = 1$ because $|B_i| > 1$.
Thus $N_S(w_i)$ contains exactly one vertex $y \in C_Y$.
\end{claimproof}

\medskip

We claim that for every pair of vertices $w_i, w_j \in S_W$, $B_i \cap B_j = \emptyset$.
Assume for contradiction that $B_i \cap B_j \neq \emptyset$.
Applying Claim~\ref{claim:insideCYandx} for $w_i$ shows that there exists $x \in I_X$ that has strong neighbours in every vertex of $B_i \cap B_j$.
With a similar argument for $w_j$ we deduce that there exists $x' \in I_X$ that has strong neighbours in every vertex of $B_i \cap B_j$.
If $x \neq x'$ then a vertex from $B_i \cap B_j$ has two distinct strong neighbours in $I_X$ which is not possible due to the strong triadic closure.
Thus $x = x'$.
Claim~\ref{claim:yCy} implies that the unique non-neighbour $y$ of $x$ is strongly adjacent to both $w_i$ and $w_j$.
This however violates the strong triadic closure for the edges of $E_S$ since $w_i,w_j$ are non-adjacent and we reach a contradiction.
Thus $B_i \cap B_j = \emptyset$.
This means that the number of edges in $E_S$ is at least $n(2n-1)+\lfloor\frac{n}{2}\rfloor+\lceil\frac{|S_W|}{2}\rceil$ which is maximized for $k=|S_W|$.
Therefore $E_S$ contains the maximum number of $|S_W|$ which is a solution for {\sc MaxDisjointNN} on $G$, since $G$ is an induced subgraph of $G'$.
\end{proof}


\section{Computing MaxSTC on proper interval graphs}
A graph is a \emph{proper interval graph} if there is a bijection between its vertices and a
family of closed intervals of the real line such that two vertices are adjacent if and
only if the two corresponding intervals overlap and no interval is properly contained in another interval.
A vertex ordering $\sigma$ is a linear arrangement $\sigma = \langle v_1, \ldots, v_n \rangle$ of the vertices of $G$.
For a vertex pair $x,y$ we write $x \preceq y$ if $x = v_i$ and $y = v_j$ for some indices $i \leq j$; if $x \neq y$ which implies $i <j$ then we write $x \prec y$.
The first position in $\sigma$ will be referred to as the {\it left end} of $\sigma$, and the last position as the {\it right end}.
We will use the expressions {\it to the left of}, {\it to the right of}, {\it leftmost}, and {\it rightmost} accordingly.

A vertex ordering $\sigma$ for $G$ is called a {\it proper interval ordering} if for every vertex triple $x,y,z$ of $G$ with $x \prec y \prec z$, $\{x,z\} \in E(G)$ implies $\{x,y\}, \{y,z\} \in E(G)$.
Proper interval graphs are characterized as the graphs that admit such orderings, that is, a graph is a proper interval graph if and only if it has a proper interval ordering \cite{LO934}.
We only consider this vertex ordering characterization for proper interval graphs.
Moreover it can be decided in linear time whether a given graph is a proper interval graph, and if so, a proper interval ordering can be generated in linear time \cite{LO934}.
It is clear that a vertex ordering $\sigma$ for $G$ is a proper interval ordering if and only if the reverse of $\sigma$ is a proper interval ordering.
A connected proper interval graph without twin vertices has a unique proper interval ordering $\sigma$ up to reversal \cite{DHH96,Iba09}.
Figure~\ref{fig:proper} shows a proper interval graph with its proper interval ordering.

\begin{figure}[t]
\centering
\includegraphics[scale=1.00]{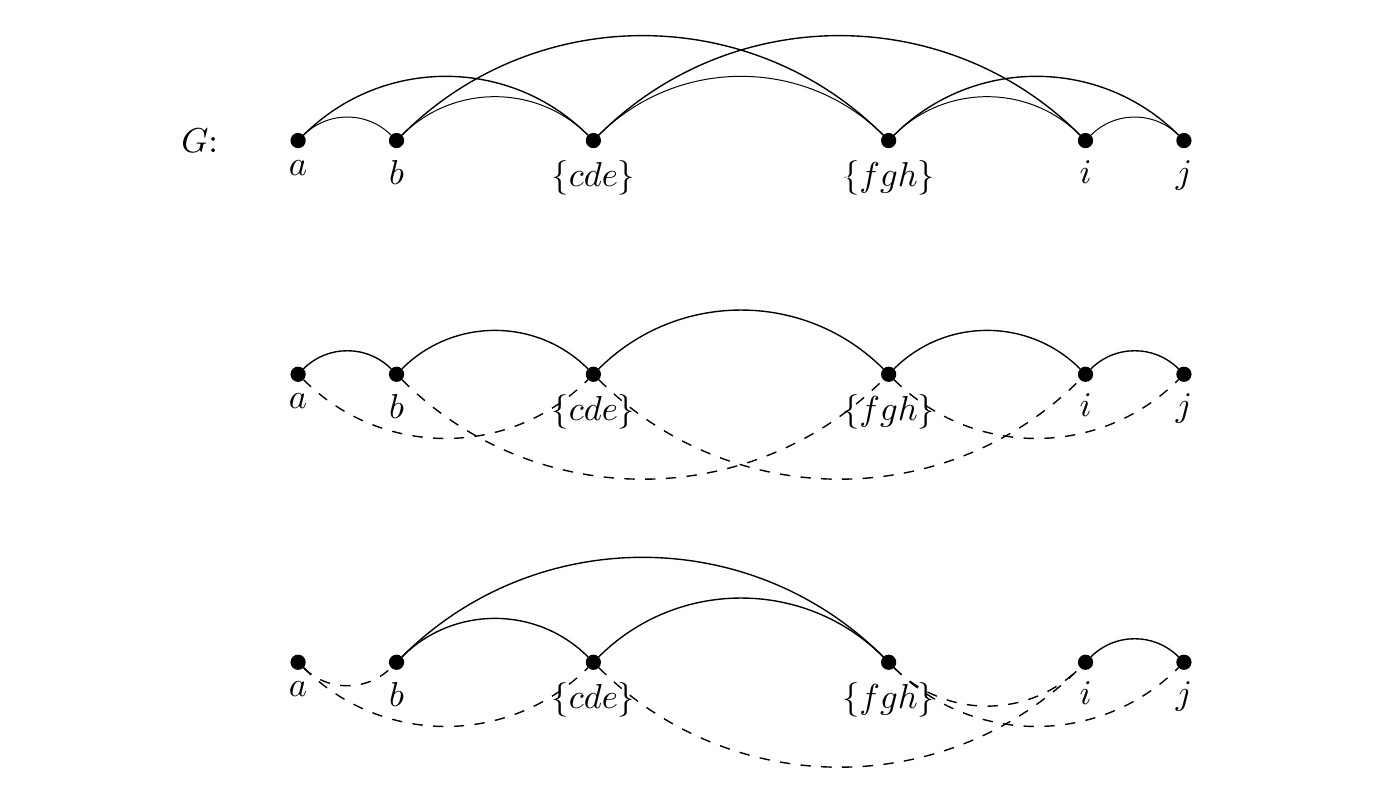}
\caption{A proper interval graph $G$ and its proper interval ordering. The vertices $\{c,d,e\}$ and $\{f,g,h\}$ form twin sets in $G$.
The two lower orderings show two solutions for {\sc MaxSTC} on $G$.
A solid edge corresponds to a strong edge, whereas a dashed edge corresponds to a weak edge.
Observe that the upper solution contains larger number of strong edges than the lower one.
Also note that the lower solution consists an optimal solution for the {\sc Cluster Deletion} problem on $G$.}\label{fig:proper}
\end{figure}

Let us turn our attention to the {\sc MaxSTC} problem.
Instead of maximizing the strong edges of the original graph $G$, we will look at the maximum independent set of the following graph that we call the \emph{line-incompatibility} graph $\widehat{G}$ of $G$:
for every edge of $G$ there is a node in $\widehat{G}$ and two nodes of $\widehat{G}$ are adjacent if and only if the vertices of the corresponding edges induce a $P_3$ in $G$.
Note that the line-incompatibility graph of $G$ is a subgraph of the line graph\footnote{The \emph{line graph} of $G$ is the graph having the edges of $G$ as vertices and two vertices of the line graph are adjacent if and only if the two original edges are incident in $G$.} of $G$.

\begin{proposition}\label{prop:MaxSTCIndependent}
A subset $S$ of edges $E(G)$ is an optimal solution for {\sc MaxSTC} of $G$ if and only if $S$ is a maximum independent set of $\widehat{G}$.
\end{proposition}
\begin{proof}
By Observation~\ref{obs:stcP3Clique} for every $P_3$ in $G$ at least one of its two edges must be labeled weak in $S$.
This means that these two edges are adjacent in $\widehat{G}$ and they cannot belong to an independent set of $\widehat{G}$.
On the other hand, by construction two nodes of $\widehat{G}$ are adjacent if and only if there is a $P_3$ in $G$.
Thus the nodes of an independent set of $\widehat{G}$ can be labeled strong in $G$ satisfying the strong triadic closure.
\end{proof}

Therefore we seek for the optimal solution of $G$ by looking at a solution for a maximum independent set of $\widehat{G}$.
As a byproduct, if we are interested in minimizing the number of weak edges then we ask for the minimum vertex cover of $\widehat{G}$.
We denote by $I_{\widehat{G}}$ the maximum independent set of $\widehat{G}$.
To distinguish the vertices of $\widehat{G}$ with those of $G$ we refer to the former as nodes and to the latter as vertices.
For an edge $\{u,v\}$ of $G$ we denote by $uv$ the corresponding node of $\widehat{G}$.
Figure~\ref{fig:incompatibility} shows the line-incompatibility graph of the proper interval graph given in Figure~\ref{fig:proper}.

\begin{figure}[t]
\centering
\includegraphics[scale=1.00]{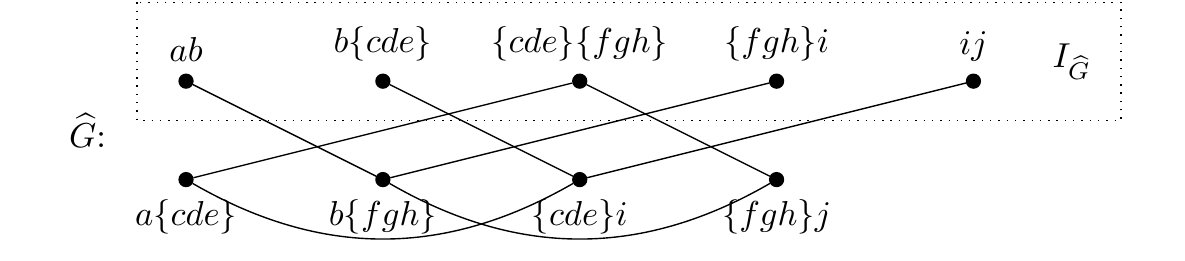}
\caption{The line-incompatibility graph $\widehat{G}$ of the proper interval graph $G$ given in Figure~\ref{fig:proper}.
The set $I_{\widehat{G}}$ is a maximum weighted independent set of $\widehat{G}$, by taking into account the weight of each node (i.e., an edge of $G$) that corresponds to the number of the twin vertices of its endpoints in $G$ (see Lemma~\ref{lem:twins}).}\label{fig:incompatibility}
\end{figure}

\begin{lemma}\label{lem:twins}
Let $x$ and $y$ be twin vertices of a graph $G$. Then there is an optimal solution $I_{\widehat{G}}$ such that $xy \in I_{\widehat{G}}$ and
for every vertex $u \in N(x)$, $xu \in I_{\widehat{G}}$ if and only if $yu \in I_{\widehat{G}}$.
\end{lemma}
\begin{proof}
First we show that $xy$ is an isolated node in $I_{\widehat{G}}$.
If $xy$ is adjacent to $xu$ then $y$ is non-adjacent to $u$ in $G$ which contradicts the fact that $x$ and $y$ are twins.
Thus $xy$ is an isolated node in $\widehat{G}$ which implies $xy \in I_{\widehat{G}}$.
For the second argument observe that for every vertex $u \in N(x)$, $xu$ and $yu$ are non-adjacent in $I_{\widehat{G}}$.
Let $u \in N(x)$. Then notice that $u \in N(y)$.
This means that if $xu \in I_{\widehat{G}}$ (resp., $yu \in I_{\widehat{G}}$) then $yu$ (resp., $xu$) is a node of $\widehat{G}$.
We define the following sets of nodes in $\widehat{G}$:
\begin{itemize}
\item Let $A_x$ be the set of nodes $xa$ such that $xa \in I_{\widehat{G}}$ and $ya \notin I_{\widehat{G}}$ and let $A_y$ be the set of nodes $ya$ such that $xa \in A_x$.
\item Let $B_y$ be the set of nodes $yb$ such that $yb \in I_{\widehat{G}}$ and $xb \notin I_{\widehat{G}}$ and let $B_x$ be the set of nodes $xb$ such that $yb \in B_y$.
\end{itemize}
It is clear that $A_x \subseteq I_{\widehat{G}}$, $B_y \subseteq I_{\widehat{G}}$, and $A_x \cap B_y = \emptyset$.
Also note that $|A_x| = |A_y|$ and $|B_y| = |B_x|$, since $N[x]=N[y]$.

Let $I_{xy} = I_{\widehat{G}} \setminus \left(A_x \cup B_y\right)$ so that $I_{\widehat{G}} = A_x \cup B_y \cup I_{xy}$.
We show that every node of $A_y$ is non-adjacent to any node of $I_{\widehat{G}} \setminus B_y$.
Let $ya$ be a node of $A_y$. 
If there is a node $az \in I_{\widehat{G}} \setminus B_y$ that is adjacent to $ya$ then $z$ and $y$ are non-adjacent in $G$ which implies that $z$ and $x$ are non-adjacent in $G$.
This however leads to a contradiction because $xa,az \in I_{\widehat{G}}$ and $xa$ is adjacent to $az$ in $\widehat{G}$.
If there is a node $yb \in I_{\widehat{G}}$ that is adjacent to $ya$ then $a$ is non-adjacent to $b$ in $G$ so that $xa$ is also adjacent to $xb$ in $\widehat{G}$. 
This means that $xb \notin I_{\widehat{G}}$ implying that $yb \in B_y$.  
Thus every node of $A_y$ is non-adjacent to any node of $I_{\widehat{G}} \setminus B_y$ and with completely symmetric arguments,
every node of $B_x$ is non-adjacent to any node of $I_{\widehat{G}} \setminus A_x$.
Hence both sets $I_1 = A_x \cup A_y \cup I_{xy}$ and $I_2 = B_x \cup B_y \cup I_{xy}$ form independent sets in $\widehat{G}$.
By the facts that $|A_x| = |A_y|$ and $|B_y| = |B_x|$ we have $|I_1| \geq |I_{xy}|$ whenever $|A_x| \geq |B_y|$ and $|I_2| \geq |I_{xy}|$ whenever $|A_x| < |B_y|$.
Therefore we can safely replace one of the sets $A_x$ or $B_y$ by $B_x$ or $A_y$ and obtain the solutions $I_2$ or $I_1$, respectively.
Now observe that in both solutions $I_1$ and $I_2$ we have $xu \in I_{i}$ if and only if $yu \in I_{i}$, for $i\in \{1,2\}$, and this completes the proof.
\end{proof}

Lemma~\ref{lem:twins} suggests to consider a graph $G$ that has no twin vertices as follows.
We partition $V(G)$ into sets of twins.
A vertex that has no twin appears in its twin set alone.
For every twin set $W_x$ we choose an arbitrary vertex $x$ and remove all its twin vertices except $x$ from $G$. 
Let $G'$ be the resulting graph that has no twin vertices. 
For every edge $\{x,y\}$ of $G'$ we assign a weight equal to the product $|W_x| \cdot |W_y|$.
This value corresponds to all edges of the original graph $G$ between the vertices of $W_x$ and $W_y$.
The line-incompatibility graph $\widehat{G'}$ of $G'$ is constructed as defined above with the only difference that a node of $\widehat{G'}$ has weight equal to the weight of its corresponding edge in $G'$.
Let $I_{\widehat{G'}}$ be a {\it maximum weighted independent set}, that is an independent set of $\widehat{G'}$ such that the sum of the weights of its nodes is maximized.
Then by Lemma~\ref{lem:twins} we have $I_{\widehat{G}} = I_{\widehat{G'}} \cup S(W)$ where $S(W)$ contains $|W_x|(|W_x|-1)/2$ nodes for every twin set $W_x$.
Therefore we are interested in computing the maximum weighted independent set of $\widehat{G'}$.
Also note that $G'$ is an induced subgraph of the original graph $G$.
Since there is no ambiguity in the forthcoming results and in order to avoid heavier notation we refer to $\widehat{G'}$ as $\widehat{G}$ by assuming that $G$ has no twin vertices and every vertex of $G$ has a positive weight.

Before reaching the details of our algorithm for proper interval graphs, let us highlight the difference between the optimal solution for {\sc MaxSTC} and the optimal solution for the {\sc Cluster Deletion}.
As already explained in the Introduction a solution for {\sc Cluster Deletion} fulfils the strong triadic closure, though the converse is not necessarily true.
In fact such an observation carries out for the class of proper interval graphs as shown in the example given in Figure~\ref{fig:proper}.
For the {\sc Cluster Deletion} problem twin vertices can be grouped together following a similar characterization with Lemma~\ref{lem:twins}, as proved in \cite{BDM15}.
Therefore when restricted to proper interval graphs the optimal solution for {\sc Cluster Deletion} does not necessarily imply an optimal solution for {\sc MaxSTC}.

Let $G$ be a proper interval graph and let $\sigma$ be a proper interval ordering for $G$.
We say that a solution $I_{\widehat{G}}$ has the \emph{consecutive strong ordering} with respect to $\sigma$ if for any three vertices $x,y,z$ of $G$ with $x \prec y \prec z$ the following holds: $xz \in I_{\widehat{G}}$ implies $xy, yz \in I_{\widehat{G}}$.
Our task is to show that such an optimal ordering exists.
We start by characterizing the optimal solution $I_{\widehat{G}}$ with respect to the proper interval ordering $\sigma$.

\begin{lemma}\label{lem:ordernoedge}
Let $x,y,z$ be three vertices of a proper interval graph $G$ such that $x \prec y \prec z$.
If $xz \in I_{\widehat{G}}$ then $xy \in I_{\widehat{G}}$ or $yz \in I_{\widehat{G}}$.
\end{lemma}
\begin{proof}
We show that at least one of $xy$ or $yz$ belongs to $I_{\widehat{G}}$.
Consider the node $xy$ in $\widehat{G}$.
If $xy$ is adjacent to a node $xx_{\ell} \in I_{\widehat{G}}$ then $\{x_{\ell},y\} \notin E(G)$.
Then observe that $x_{\ell} \prec y$ because $x \prec y$ and $\{x_{\ell},y\} \notin E(G)$.
Since both $xx_{\ell}$ and $xz$ belong to $I_{\widehat{G}}$, $\{x_{\ell},z\} \in E(G)$.
This however contradicts the proper interval ordering because $x_{\ell} \prec y \prec z$, $\{x_{\ell},z\} \in E(G)$ and $y$ is non-adjacent to $x_{\ell}$.
Thus $xy$ is non-adjacent to any node $xx_{\ell} \in I_{\widehat{G}}$ and, in analogous fashion, $yz$ is non-adjacent to any node $zz_{r} \in I_{\widehat{G}}$.

Now assume that $xy$ is adjacent to a node $yy_{r} \in I_{\widehat{G}}$ and $yz$ is adjacent to a node $y_{\ell}y \in I_{\widehat{G}}$.
This means that $\{x,y_{r}\} \notin E(G)$ and $\{z,y_{\ell}\} \notin E(G)$.
Since $\{x,z\} \in E(G)$, by the proper interval ordering we have $y_{\ell} \prec x \prec y \prec z \prec y_r$.
Then notice that $\{y_{\ell}, y_r\} \in E(G)$, because both $yy_{r}, yy_{\ell} \in I_{\widehat{G}}$.
By the proper interval ordering we know that both $x$ and $z$ are adjacent to $y_{\ell}, y_r$, leading to a contradiction to the assumptions $\{x,y_{r}\} \notin E(G)$ and $\{z,y_{\ell}\} \notin E(G)$.
Therefore at least one of $xy$ or $yz$ belongs to $I_{\widehat{G}}$.
\end{proof}

Thus by Lemma~\ref{lem:ordernoedge} we have two symmetric cases to consider.
The next characterization suggests that there is a fourth vertex with important properties in each corresponding case.

\begin{lemma}\label{lem:ordernow}
Let $x,y,z$ be three vertices of a proper interval graph $G$ such that $x \prec y \prec z$ and $xz \in I_{\widehat{G}}$.
\begin{itemize}
\item If $xy \notin I_{\widehat{G}}$ and $yz \in I_{\widehat{G}}$ then $xy$ is non-adjacent to any node $x_{\ell}x \in I_{\widehat{G}}$ and there is a vertex $w$ such that $yw \in I_{\widehat{G}}$, $\{x,w\} \notin E(G)$, and $z \prec w$.
\item If $xy \in I_{\widehat{G}}$ and $yz \notin I_{\widehat{G}}$ then $yz$ is non-adjacent to any node $zz_{r} \in I_{\widehat{G}}$ and there is a vertex $w$ such that $wy \in I_{\widehat{G}}$, $\{w,z\} \notin E(G)$ and $w \prec x$.
\end{itemize}
\end{lemma}
\begin{proof}
Let $xy \notin I_{\widehat{G}}$ and $yz \in I_{\widehat{G}}$.
The case for $xy \in I_{\widehat{G}}$ and $yz \notin I_{\widehat{G}}$ is completely symmetric.
Assume for contradiction that there is no vertex $w$ such that $yw \in I_{\widehat{G}}$, $\{x,w\} \notin E(G)$, and $z \prec w$.
We prove that $xy$ is non-adjacent to any node of $I_{\widehat{G}}$, contradicting the optimality of $I_{\widehat{G}}$.
Suppose first that $xy$ is adjacent to a node $x_{\ell}x \in I_{\widehat{G}}$.
Then $y$ is non-adjacent to $x_{\ell}$ in $G$.
Notice that $x_{\ell} \prec x$ because $y$ is adjacent to $x$ and $x \prec y$.
Due to the fact that $yz \in I_{\widehat{G}}$, we have that $x_{\ell}x$ and $xz$ are non-adjacent in $\widehat{G}$ which implies that $\{x_{\ell},z\} \in E(G)$.
Since $x_{\ell} \prec x \prec y \prec z$ and $\{x_{\ell},z\} \in E(G)$, by the proper interval ordering we get $\{x_{\ell},y\} \in E(G)$ leading to a contradiction.
Thus $xy$ is non-adjacent to any node $x_{\ell}x \in I_{\widehat{G}}$.

Next assume that $xy$ is adjacent to a node $yy_{r} \in I_{\widehat{G}}$.
Then $\{x,y_r\} \notin E(G)$.
By the assumption that there is no vertex $w$ with $yw \in I_{\widehat{G}}$, $\{x,w\} \notin E(G)$, and $z \prec w$, we have $y_r \prec z$.
This particularly means that $y_r \prec x$ or $x \prec y_r \prec z$.
However both cases lead to a contradiction to $\{x,y_r\} \notin E(G)$ since in the former case we have $\{y_r,y\} \in E(G)$ and $y_r \prec x \prec y$, and in the latter case we know that $\{x,z\}\in E(G)$.
Therefore $xy$ has no neighbor in $I_{\widehat{G}}$ reaching a contradiction to the optimality of $I_{\widehat{G}}$.
\end{proof}

Now we are ready to show that that there is an optimal solution that has the described properties with respect to the given proper interval ordering.

\begin{lemma}\label{lem:consecutivestrong}
There exists an optimal solution $I_{\widehat{G}}$ that has the consecutive strong property with respect to $\sigma$.
\end{lemma}
\begin{proof}
Let $\sigma$ be a proper interval ordering for $G$.
Assume for contradiction that $I_{\widehat{G}}$ does not have the consecutive strong property.
Then there exists a {\it conflict} with respect to $\sigma$, that is, there are three vertices $x,y,z$ with $x \prec y \prec z$ and $xz \in I_{\widehat{G}}$ such that $xy \notin I_{\widehat{G}}$ or $yz \notin I_{\widehat{G}}$.
We will show that as long as there are conflicts in $\sigma$, we can reduce the number of conflicts in $\sigma$ without affecting the value of the optimal solution $I_{\widehat{G}}$.
Consider such a conflict formed by the three vertices $x \prec y \prec z$ with $xz \in I_{\widehat{G}}$.
By Lemma~\ref{lem:ordernoedge} we know that $xy \in I_{\widehat{G}}$ or $yz \in I_{\widehat{G}}$.
Assume that $yz \in I_{\widehat{G}}$.
Then clearly $xy \notin I_{\widehat{G}}$, for otherwise there is no conflict.
Then by Lemma~\ref{lem:ordernow} there is a vertex $w$ such that $yw \in I_{\widehat{G}}$, $\{x,w\} \notin E(G)$, and $x \prec y \prec z \prec w$.
Notice that both triples $x,y,z$ and $y,z,w$ create conflicts in $\sigma$.

We start by choosing an appropriate such conflict that is formed by four vertices $x,y,z,w$ so that $x \prec y \prec z \prec w$, $xz,yz,yw \in I_{\widehat{G}}$, and $\{x,w\} \notin E(G)$.
Fix $y$ and $z$ in $\sigma$ with $y,z$ being the leftmost and the rightmost vertices, respectively, such that for every vertex $v$ with $y \prec v \prec z$, $yv,vz \in I_{\widehat{G}}$ holds.
Recall that $yz \in I_{\widehat{G}}$.
We choose $x$ as the leftmost vertex such that $xz \in I_{\widehat{G}}$ and we choose $w$ as the rightmost vertex such that $yw \in I_{\widehat{G}}$.
Observe that $\{x,w\} \notin E(G)$ since $y$ and $z$ participate in a conflict.
Due to the properties of the considered conflicts all such vertices exist (see for e.g., Figure~\ref{fig:orderings}).

Let $W(x)$ be the set of vertices $w_i$ such that $yw_i \in I_{\widehat{G}}$ and $\{x,w_i\} \notin E(G)$, and
let $X(w)$ be the set of vertices $x_j$ such that $x_jz \in I_{\widehat{G}}$ and $\{x_j,w\} \notin E(G)$.
Observe that $w \in W(x)$ and $x \in X(w)$ such that $w$ is the rightmost vertex in $W(x)$ and $x$ is the leftmost vertex in $X(w)$.
For a vertex $w_i$ of $W(x)$ observe the following.
If $w_i \prec x$ then $\{w_i,x\} \in E(G)$ because $\{w_i,y\}\in E(G)$ and if $x \prec w_i \prec z$ then $\{w_i,x\} \in E(G)$ because $\{x,z\}\in E(G)$.
Thus $z \prec w_i$ which implies that $\{z,w_i\} \in E(G)$ since $\{y,w_i\}\in E(G)$.
If $zw_i \in I_{\widehat{G}}$ then by the fact that $xz \in I_{\widehat{G}}$ we have $\{x,w_i\} \in E(G)$ contradicting the definition of $W(x)$.
Moreover for every vertex $b_1$ such that $w_ib_1 \in I_{\widehat{G}}$ notice that $x \prec b_1$ since $\{x,w_i\} \notin E(G)$.
If $x \prec b_1 \prec w_i$ then $\{z,b_1\} \in E(G)$ since $x \prec z \prec w_i$;
and if $w_i \prec b_1$ then due to the fact that $yw_i, w_ib_1 \in I_{\widehat{G}}$ and $\{y,b_1\} \in E(G)$ we have again $\{z,b_1\} \in E(G)$ since $y \prec z \prec b_1$.
Furthermore consider a vertex $b_2$ such that $z \prec b_2 \prec w$ and $b_2 \notin W(x)$.
This means that $yb_2 \notin I_{\widehat{G}}$ or $yb_2 \in I_{\widehat{G}}$ with $\{b_2,x\}\in E(G)$.
The latter case implies that $b_2$ is adjacent to every vertex of $X(w)$, since $x$ is the leftmost vertex in $X(w)$ and every vertex of $X(w)$ is to the left of $z$.
Hence for every vertex $w_i$ of $W(x)$ the following hold:
\begin{itemize}
\item $z \prec w_i$,
\item $zw_i \notin I_{\widehat{G}}$,
\item for every node $w_ib_1 \in I_{\widehat{G}}$, $\{z,b_1\} \in E(G)$, and
\item for every vertex $b_2$ with $z \prec b_2 \prec w$ and $b_2 \notin W(x)$, $yb_2 \notin I_{\widehat{G}}$ or $b_2$ is adjacent to every vertex of $X(w)$.
\end{itemize}
With symmetric arguments for every vertex $x_j$ of $X(w)$ we have the following:
\begin{itemize}
\item $x_j \prec y$,
\item $x_jy \notin I_{\widehat{G}}$,
\item for every node $a_1x_j \in I_{\widehat{G}}$, $\{a_1,y\} \in E(G)$, and
\item for every vertex $a_2$ with $x \prec a_2 \prec y$ and $a_2 \notin X(w)$, $a_2z \notin I_{\widehat{G}}$ or $a_2$ is adjacent to every vertex of $W(x)$.
\end{itemize}
The topmost ordering given in Figure~\ref{fig:orderings} illustrates the corresponding cases.

\begin{figure}[t]
\centering
\includegraphics[width=\textwidth]{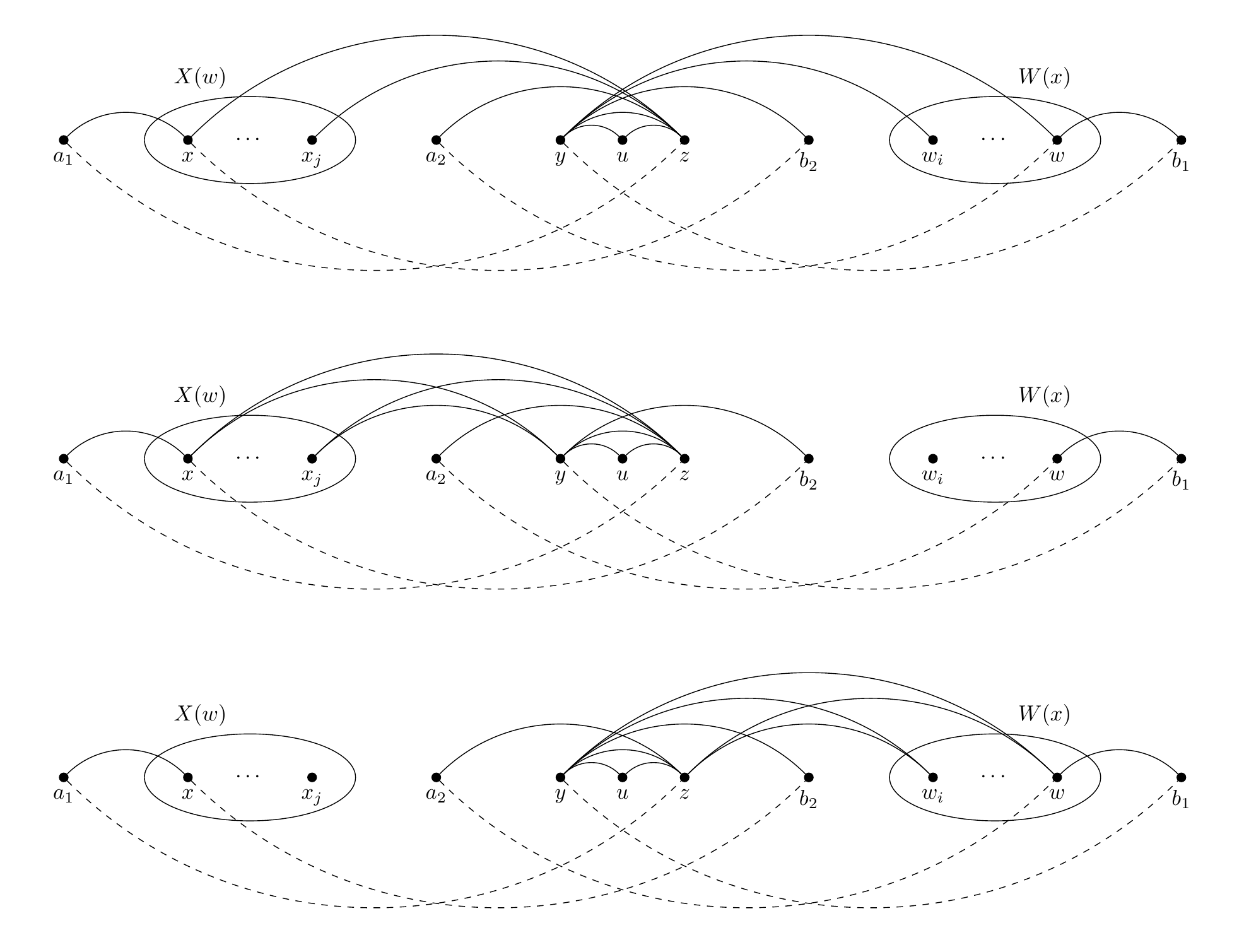}
\caption{A proper interval ordering for a graph $G$ with three different solutions considered in the proof of Lemma~\ref{lem:consecutivestrong}.
A solid edge corresponds to a node of $\widehat{G}$ that belongs to $I_{\widehat{G}}$, which means that such an edge is labeled strong in an optimal strong-weak labeling,
whereas a dashed edge corresponds to a node of $\widehat{G}$ that does not belong to $I_{\widehat{G}}$, which means that such an edge is labeled weak in an optimal strong-weak labeling.
Observe that the lowest two orderings contain less {\it conflicts} than the topmost, that is, triple of vertices that violate the consecutive strong property.}\label{fig:orderings}
\end{figure}

Let $Y_w$ be the set of nodes $yw_i$ in $\widehat{G}$ such that $w_i \in W(x)$, and let $Z_x$ be the set of nodes $x_jz$ in $\widehat{G}$ such that $x_j \in X(w)$.
Observe that $Y_w, Z_x \subseteq I_{\widehat{G}}$ by the previous arguments.
We show that removing either $Y_w$ or $Z_x$ from $I_{\widehat{G}}$ does not create any new conflict.
Let $yw_i \in Y_w$ and let $u$ be a vertex such that $uy \in I_{\widehat{G}}$ and $uw_i \in I_{\widehat{G}}$.
If $y \prec u \prec w_i$ then no conflict is created by removing $yw_i$ from $I_{\widehat{G}}$.
Assume that $u \prec y \prec w_i$. Observe that $x \prec u \prec z$.
Then $xu \notin I_{\widehat{G}}$ because $\{x,w_i\} \notin E(G)$.
Since $xz \in I_{\widehat{G}}$ and at least one of $xu$, $uz$ belongs to $I_{\widehat{G}}$, we have $uz \in I_{\widehat{G}}$.
However this contradicts the leftmost choice for $y$ in $x \prec u \prec y \prec z$ and there is no such vertex $u$.
Next assume that $y \prec w_i \prec u$. Since $w_i$ is non-adjacent to $x$ and $w_i \prec u$, $u$ is non-adjacent to $x$, as well.
Then according to the definition of $W(x)$, $u \in W(x)$ and $yu \in Y_w$.
The case for the nodes of $Z_x$ is completely symmetric.
Thus no conflicts are created by removing the nodes of $Y_w$ or the nodes of $Z_x$ from $I_{\widehat{G}}$.

Let $Y_x$ be the set of nodes $x_jy$ in $\widehat{G}$ such that $x_j \in X(w)$, and let $Z_w$ be the set of nodes $zw_i$ in $\widehat{G}$ such that $w_i \in W(x)$.
We denote by $I(Y_x)$ and $I(Z_w)$ the following sets of nodes:
$I(Y_x) = \left(I_{\widehat{G}} \setminus Y_w\right) \cup Y_x$ and
$I(Z_w) = \left(I_{\widehat{G}} \setminus Z_x\right) \cup Z_w$.
We show that both sets form independent sets in $\widehat{G}$.
Consider the case for $I(Y_x)$.
Let $x_jy$ be a node of $I(Y_x)$. There are two cases to consider:
there is a node $ux_j \in \left(I_{\widehat{G}} \setminus Y_w\right) \cup Y_x$ and $\{u,y\} \notin E(G)$ or
there is a node $yv \in \left(I_{\widehat{G}} \setminus Y_w\right) \cup Y_x$ and $\{x_j,v\} \notin E(G)$.
In the former case we know that $ux_j \in I_{\widehat{G}}$ and for every node $ux_j \in I_{\widehat{G}}$, $\{u,y\}$ must be an edge of $G$ which leads to a contradiction to the non-adjacency of $u$ and $y$.
For the latter case observe that $yv \in I_{\widehat{G}} \setminus Y_w$ and $v \notin W(x)$.
Since $\{x_j,v\} \notin E(G)$ and $\{y,v\} \in E(G)$, we have $z \prec v$ and by the rightmost choice of $w$ for $y$ we have $z \prec v \prec w$.
This however implies that $z \prec v \prec w$, $v \notin W(x)$ and $yv \in I_{\widehat{G}}$ showing that $\{x_j,v\} \in E(G)$.
Completely symmetric arguments hold for $I(Z_w)$.
The two lowest orderings given in Figure~\ref{fig:orderings} illustrate the considered cases.
Thus $I(Y_x)$ and $I(Z_w)$ form independent sets in $\widehat{G}$.

Now observe that both $I(Y_x)$ and $I(Z_w)$ have a smaller number of conflicts with respect to $I_{\widehat{G}}$ because either $x,y,z$ in $I(Y_x)$ or $y,z,w$ in $I(Z_w)$ satisfy the consecutive strong property.
It is clear that the difference between $I(Y_x)$ and $I_{\widehat{G}}$ are the nodes of $Y_x$ and $Y_w$, whereas the difference between $I(Z_w)$ and $I_{\widehat{G}}$ are the nodes of $Z_w$ and $Z_x$.
For a set $A$ of vertices having positive weights, denote by $M(A)$ the sum of the weights of its vertices.
If $M(X(w)) \geq M(Z(x))$ then $M(I(Y_x)) \geq M(I_{\widehat{G}})$ and if $M(X(w)) < M(Z(x))$ then $M(I(Z_w)) > M(I_{\widehat{G}})$.
Thus in any case we can replace appropriate set of nodes in $I_{\widehat{G}}$ and obtain an optimal solution with a smaller number of conflicts.
Therefore by applying such a replacement in every such conflict, we get an optimal solution that has no conflicts and, thus, it satisfies the consecutive strong property.
\end{proof}

Lemma~\ref{lem:consecutivestrong} suggests to find an optimal solution that has the consecutive strong property with respect to $\sigma$.
In fact by Proposition~\ref{prop:MaxSTCIndependent} and the proper interval ordering, this reduces to computing the largest proper interval subgraph $H$ of $G$ such that the vertices of every $P_3$ of $H$ induce a clique in $G$.

Let $G$ be a proper interval graph and let $\sigma = \langle v_1, \ldots, v_n \rangle$ be its proper interval ordering.
For a vertex $v_i$ we denote by $v_{\ell(i)}$ and $v_{r(i)}$ its leftmost and rightmost neighbors, respectively, in $\sigma$.
Observe that for any two vertices $v_i \prec v_j$ in $\sigma$, $v_{\ell(i)} \preceq v_{\ell(j)}$ and $v_{r(i)} \preceq v_{r(j)}$ \cite{DHH96}.
For $1 \leq i \leq r(1)$, let $V_i = \{v_1, \ldots, v_i\}$, that is, $V_i$ contains the {\it first} $i$ vertices in $\sigma$.
Given the set $V_i$, let $r$ be an integer such that $1 \leq i \leq r \leq r(1)$.

Let $A(G)$ be the value of an optimal solution $I_{\widehat{G}}$ for $G$.
For the set $V_i$ we denote by $B(V_i)$ the value that corresponds to the total weight of the edges incident to $v_1$ and each of $v_2, \ldots, v_i$.
Observe that any subset of vertices of $V_i$ induces a clique in $G$.
We denote by $C(V_i)$ the value that corresponds to the total weight of the edges among all vertices of $V_i$. 
Given the first vertices $V_i$ and $i \leq r \leq r(1)$, let $A(G,V_i, r)$ be the value of the optimal solution $I(G,V_i, r)$ in $G$ such that every edge among the vertices of $V_i$ belongs to $I(G,V_i, r)$ and $uv_k \notin I(G,V_i, r)$ with $u \in V_i$ and $k>r$.
As a trivial case observe that if $G$ contains exactly two vertices $v_1,v_2$ then $A(G) = B(\{v_1,v_2\}) = C(\{v_1,v_2\})$.

\begin{lemma}\label{lem:optreccurence}
Let $G$ be a proper interval graph and let $V_i$ and $r$ such that $1 \leq i \leq r \leq r(1)$. Then, $A(G) = A(G, \emptyset, r) = A(G, \{v_1\}, r(1))$, and 
\[
A(G, V_i, r) =
\begin{cases}
\max\limits_{i \leq j \leq r} \left\{ A(G-\{v_1\}, V_j \setminus \{v_1\}, r) + B(V_j) \right\} & \text{ if } i < r,\\
A(G - V_i, \{v_{i+1}\}, r(i+1)) + C(V_i)  & \text{ if } i = r \text{ and } i < n,\\
C(V_i) & \text{otherwise}. 
\end{cases}
\]
\end{lemma}
\begin{proof}
We show that $A(G)$ computes the value of an optimal solution that satisfies the consecutive strong property.
By Lemma~\ref{lem:consecutivestrong} such an ordering exists.
Since there is no edge between $v_1$ and $v_k$ with $k > r(1)$, it follows that $A(G) = A(G, \{v_1\}, r(1))$.
Observe that every induced subgraph of a proper interval graph is proper interval, which implies that the graphs $G - \{v_1\}$ and $G - V_i$ remain proper interval.
Let $V_i= \{v_1, \ldots, v_i\}$ be the first $i$ vertices of $G$.
According to Lemma~\ref{lem:consecutivestrong} if $v_iv_k \notin I(G,V_i,r)$ with $k>i$, then every node $v_jv_k$, $1 \leq j \leq i$, does not belong to $I(G,V_i,r)$.
Again by Lemma~\ref{lem:consecutivestrong} if $v_1v_k \in I(G,V_i,r)$ with $k>i$, then every node $v_jv_k$, $1 \leq j \leq i$, belongs to $I(G,V_i,r)$.

Let $I(V_i)$ be the nodes of $\widehat{G}$ corresponding to the edges of $G[V_i]$.
Observe that by the definition of $I(G, V_i, r)$, $I(V_i) \subseteq I(G, V_i, r)$.
If $i=r$ then it is clear that the nodes of $I(V_i)$ belong to $I(G, V_i, r)$.
Moreover if $i = r < n$ then no node $v_iv_k$, with $i<k$, can be added in $I(G, V_i, r)$, implying that all the nodes corresponding to the edges between $V_i$ and $G - V_i$ do not belong to $I(G, V_i, r)$.
Thus we can safely add $I(V_i)$ to an optimal solution for $G - V_i$, that is, $I(G - V_i, \{v_{i+1}\}, r(i+1))$.

Assume that $i < r$.
To see that $I(V_i) \subseteq I(G, V_i, r)$, notice that $V_i \subseteq V_j$ for every $i \leq j \leq r$ and $C(V_i) = B(V_i) + B(V_{i}\setminus\{v_1\}) + \cdots + B(\{v_{i-1},v_i\})$.
Consider the nodes $v_1v_j$ of $\widehat{G}$ corresponding to the edges incident to $v_1$.
Let $v_1v_j \in I_{\widehat{G}}$ with $j$ as large as possible.
Then $i \leq j \leq r \leq r(1)$.
By Lemma~\ref{lem:consecutivestrong} every node $v_1v_{j'}$ of $\widehat{G}$ with $1 \leq j' \leq j$ belongs to $I_{\widehat{G}}$.
Thus $I(V_j) \subseteq I(G, V_i, r)$.
Furthermore for every node $v_jv_k$, with $r \leq r(1) < k$, we know that $v_jv_k$ is adjacent to $v_1v_j$ in $\widehat{G}$ since $\{v_1,v_k\} \notin E(G)$.
Thus $v_jv_k \notin I_{\widehat{G}}$ which implies that every node $v_{j'}v_{k}$ with $j' \leq j$ does not belong to $I_{\widehat{G}}$.
This particularly means that $v_jv_k \in I(G, V_i, r)$ only if $k \leq r \leq r(1)$.
Therefore given an optimal solution by removing $v_1$ for all such possible values for $j$, we obtain the desired solution.
\end{proof}

Now we are equipped with our necessary tools in order to obtain our main result, namely a polynomial-time algorithm that solves the
{\sc MaxSTC} problem on a proper interval graph $G$.

\begin{theorem}\label{theo:algoproperinterval}
There is a polynomial-time algorithm that computes the {\sc MaxSTC} of a proper interval graph.
\end{theorem}
\begin{proof}
Let $G$ be a proper interval graph on $n$ vertices and $m$ edges.
We first compute its proper interval ordering $\sigma$ in linear time \cite{LO934}.
Then we compute its twin sets by using the fact that $u$ and $v$ are twins if and only if $\ell(u)=\ell(v)$ and $r(u)=r(v)$.
Contracting the twin sets according to Lemma~\ref{lem:twins} results in a proper interval graph in which every vertex is associated with a positive weight.
In order to compute the optimal solution $A(G)$ we use a dynamic programming approach based on its recursive formulation given in Lemma~\ref{lem:optreccurence}.
Correctness follows from Proposition~\ref{prop:MaxSTCIndependent} and Lemmata~\ref{lem:consecutivestrong} and \ref{lem:optreccurence}.

Regarding the running time, notice that given the ordering $\sigma$ we can remove the twin vertices in linear time.
All instances of $A(G, V_i, r)$ can be computed as follows.
Given the first vertex $v_1$ we compute all possible sets $V_i$ which are bounded by $n$.
Since $r \leq r(1) \leq n$, the number of instances $A(G, V_i, r)$ generated by $v_1$ is at most $n^2$.
Also observe that computing the values $B(V_i)$ and $C(V_i)$ takes $O(m)$ time.
Therefore the total running time of the algorithm is $O(n^3m)$.
\end{proof}

\section{Concluding remarks}
Given the first study with positive and negative results for the {\sc MaxSTC} problem on restricted input, there are some interesting open problems.
Despite the structural properties that we proved for the solution on proper interval graphs, the complexity of {\sc MaxSTC} on interval graphs is still open.
Determining the complexity of {\sc MaxSTC} for other graph classes towards AT-free graphs seems interesting direction for future work.
More precisely, by Proposition~\ref{prop:MaxSTCIndependent} it is interesting to consider the line-incompatibility graph of comparability graphs since they admit a well-known similar characterization \cite{KrMcMeSp06}.
Moreover it is natural to characterize the graphs for which their line-incompatibility graph is perfect.
Such a characterization will lead to to further polynomial cases of {\sc MaxSTC}, since the maximum independent set of perfect graphs admits a polynomial solution \cite{GLS84}.
A typical example is the class of bipartite graphs for which their line graph coincides with their line-incompatibility graph and it is known that the line graph of a bipartite graph is perfect (see for e.g., \cite{BraLeSpi99}).
As we show next, another paradigm of this type is the class of trivially-perfect graphs.

A graph $G$ is called {\em trivially-perfect} (also known as {\em quasi-threshold}) if for each induced subgraph $H$ of $G$, the number
of maximal cliques of $H$ is equal to the maximum size of an independent set of $H$.
It is known that the class of trivially-perfect graphs coincides with the class of $(P_4,C_4)$-free graphs, that is every trivially-perfect graph has no induced $P_4$ or $C_4$ \cite{Gol78}.
A {\em cograph} is a graph without an induced $P_4$, that is a cograph is a $P_4$-free graph.
Hence trivially-perfect graphs form a subclass of cographs.

\begin{theorem}\label{theo:trivially-perfect}
The line-incompatibility graph of a trivially-perfect graph is cograph.
\end{theorem}
\begin{proof}
Let $G$ be a trivially-perfect graph, that is $G$ is a $(P_4,C_4)$-free graph.
We will show that the line-incompatibility graph $\widehat{G}$ of $G$ is a $P_4$-free graph.
Consider any $P_3$ in $\widehat{G}$.
Due to the construction of $\widehat{G}$, the $P_3$ has one of the following forms: (i) $v_1v_2, v_2v_3, v_3v_4$ or (ii) $v_1x, v_2x, v_3x$.
We prove that the $P_3$ has the second form because $G$ has no induced $P_4$ or $C_4$.
If (i) applies then $\{v_1,v_3\}, \{v_2,v_4\} \notin E(G)$ and $\{v_1,v_2\}, \{v_2,v_3\}, \{v_3,v_4\} \in E(G)$ which implies that $v_4 \neq v_1$.
Thus $G$ contains a $P_4$ or a $C_4$ depending on whether there is the edge $\{v_1,v_4\}$ in $G$.
Hence every $P_3$ in $\widehat{G}$ has the form $v_1x, v_2x, v_3x$ where $v_1, v_2, v_3, x$ are distinct vertices of $G$.
Now assume for contradiction that $\widehat{G}$ contains a $P_4$.
Then the $P_4$ is of the form $v_1x, v_2x, v_3x, v_4x$ because it contains two induced $P_3$'s.
The structure of the $P_4$ implies that $\{v_1,v_2\}, \{v_2,v_3\}, \{v_3,v_4\} \notin E(G)$ and $\{v_1,v_3\}, \{v_2,v_4\}, \{v_4,v_1\} \in E(G)$.
This however shows that the vertices $v_3,v_1,v_4,v_2$ induce a $P_4$ in $G$ leading to a contradiction that $G$ is a $(P_4,C_4)$-free graph.
Therefore $\widehat{G}$ is a $P_4$-free graph.
\end{proof}
By Theorem~\ref{theo:trivially-perfect} and the fact that the maximum independent set of a cograph can be computed in linear time \cite{CPS81}, {\sc MaxSTC} can be solved in polynomial time on trivially-perfect graphs.
We would like to note that the line-incompatibility graph of a cograph or a proper interval graph is not necessarily a perfect graph.

More general there are extensions and variations of the {\sc MaxSTC} problem that are interesting to consider as proposed in \cite{ST14}.
A natural graph modification problem is the problem of adding the minimum number of edges that satisfy the strong triadic closure having as few as possible weak edges in the original graph.
More formally the objective is to add $F$ edges in $G$ so that the resulting graph satisfies the strong triadic closure with the minimum number of $|F|+|E_W|$.




\end{document}